\documentclass{ifacconf}
\usepackage{url}
\usepackage{graphicx}      % include this line if your document contains figures
\usepackage{natbib}        % required for bibliography
\usepackage{xcolor} % Ensure this is in the preamble
\usepackage{circuitikz}
\usepackage{graphicx} % in preamble

\DeclareMathSizes{10}{10}{6}{6}
\newcommand{\R}{\mathbb{R}}

\usepackage{amsmath,amssymb}

\newtheorem{definition}{Definition}
\newtheorem{assumption}{Assumption}
\newtheorem{remark}{Remark}
\newtheorem{theorem}{Theorem}
\newenvironment{proof}[1][Proof]{%
  \par\noindent\textit{#1.}\ }{\qed\par}
\begin{document}
\begin{frontmatter}

\title{Negative Imaginary and Passivity Properties of Synchronous Machine Systems}

\author[First]{Maryam Khodabakhshloo} 
\author[second]{Elizabeth L. Ratnam} 
\author[First]{Ian R. Petersen} 
\address[First]{School of Engineering, The Australian National University, Canberra, Australia (e-mail: maryam.khodabakhshloo, ian.petersen)@anu.edu.au.}
\address[second]{Department of Electrical and Computer Systems Engineering, Monash University, Melbourne, Australia (e-mail: liz.ratnam@monash.edu).}

% \color{red}
\begin{abstract}
    The recent rapid proliferation of renewable energy is fundamentally changing the dynamic operations of power systems, necessitating new approaches to assess stability for these highly nonlinear systems. In this paper, we prove that synchronous machine systems, modeled in the nonlinear $dq$-frame, possess fundamental dissipativity properties. Specifically, we show passivity from current input to voltage output and a nonlinear negative imaginary property from torque input to rotor angle output. For the nonlinear system shifted around an equilibrium point, we derive explicit conditions for both passivity and the NI property to hold.
    Finally, we demonstrate that interconnection with passive droop controllers preserves these dissipativity properties with identical supply rates, thereby ensuring closed-loop stability.
\end{abstract}


\end{frontmatter}
%===============================================================================

\section{Introduction}
% Today’s electrical power systems are large-scale, complex, and highly nonlinear networks whose stable and reliable operation is paramount. Ensuring this stability is increasingly challenging due to inherent nonlinear couplings, external disturbances, and the growing integration of power-electronic interfaces \cite{kundur2004definition}. The core dynamics of these AC grids are governed by trigonometric interactions, making analytical stability assessment a non-trivial task.

The transition towards inverter-based renewable energy resources is introducing new dynamic behaviors and instability modes in power systems, demanding analytical methods capable of handling their highly nonlinear nature \cite{kundur2004definition}. Consequently, classical approaches to transient stability analysis, as detailed in foundational books \cite{kundur2007power, sauer2017power, machowski2020power}, rely heavily on energy-based and Lyapunov-based methods. These have been extended in recent years using concepts from passivity and distributed control theory \cite{yang2019distributed, trip2016internal}. A common thread among these studies is the use of simplified synchronous machine models, such as the swing equation or the one-axis model. These models are derived using Singular Perturbation Approximations (SPA), which neglects fast electrical transients like flux and voltage dynamics due to their shorter time constants. While this simplification is valid for studying slow electromechanical oscillations, it fails to capture the fast transients that dominate during short-term disturbances \cite{monshizadeh2018output}, limiting the models' accuracy for detailed stability analysis.

% Recent research has pursued more detailed modelling approaches. Notably, \cite{nishino2025equilibrium} utilizes the full Park model with comprehensive field and damper winding dynamics to establish equilibrium-independent passivity properties. Building upon this foundation, this paper adopts a simplified yet physically interpretable model that assumes constant field current and neglects subtransient flux dynamics. This formulation preserves essential electromechanical coupling while enhancing analytical tractability, thereby facilitating the application of energy-based passivity and nonlinear negative-imaginary (NNI) frameworks to multi-machine power systems.
% \color{red}
Recent research has pursued more detailed modeling approaches. Notably, \cite{nishino2025equilibrium} utilizes the full Park model with comprehensive field and damper winding dynamics in a linearized setting, establishing equilibrium-independent passivity properties. By contrast, our analysis is fully nonlinear with a simplified but physically interpretable model (constant field current; no subtransient flux). This formulation preserves the essential electromechanical coupling while enhancing analytical tractability, enabling energy-based nonlinear negative-imaginary (NNI) and passivity analysis. Also, \cite{arghir2018grid} prove  incremental passivity for a converter system using a similar approach to the shifted passivity result presented in this paper. 
\color{black}

%  By contrast, our analysis is fully nonlinear with a simplified but physically interpretable model (constant field current; no subtransient flux). This formulation preserves the essential electromechanical coupling while enhancing analytical tractability, enabling energy-based passivity and nonlinear negative-imaginary (NNI) analysis for multi-machine systems.
% \color{red}
% Complementary to our focus on the dissipativity of synchronous machines themselves, recent advances in grid-forming control have sought to emulate these machine dynamics in power converters to stabilize low-inertia grids. For instance, \cite{arghir2018grid} propose a matching control design that augments inverters with a virtual oscillator to replicate synchronous machine behavior, including strict incremental passivity, droop characteristics, and power-sharing properties. Our dissipativity analysis of synchronous machines provides a foundational understanding that can inform and validate such emulation-based control designs, particularly in ensuring stability under interconnection and disturbances.
% \color{black}

In this work, we consider a nonlinear model of a synchronous machine and analyze its energy-based properties using passivity and nonlinear negative-imaginary (NNI) theory.
First, we prove that the system is passive from the current input to the voltage output. This analysis is then extended to a shifted equilibrium model, for which we derive explicit conditions on the system parameters and steady-state operating points that guarantee the passivity of the shifted system.
Second, we show that the system exhibits a nonlinear negative-imaginary property from the torque input to the rotor angle output and that this property is also preserved for the shifted system under specific parameter conditions.
Furthermore, we analyze the interconnection with passive droop controllers and demonstrate that passivity is preserved under identical supply rates. Together, these analytical results reveal explicit relationships between physical parameters and stability guarantees, providing a theoretical foundation for assessing the stability of modern power systems.
This analysis bridges classical reduced-order models and full Park representations, enabling tractable, energy-based assessment of low-inertia power networks while maintaining physical interpretability.

The remainder of the paper is organized as follows. Section~\ref{sec:prelim} reviews preliminaries on passivity and nonlinear negative-imaginary (NNI) systems. Section~\ref{sec:modeling} presents the synchronous machine system model,
Section~\ref{sec:main} provides the main passivity and NNI results for both the original and shifted systems and analyzes their interconnection with a passive/NI droop controller. Section~\ref{sec:conclusion} concludes the paper.

\section{Preliminaries}\label{sec:prelim}
This section reviews fundamental dissipativity concepts \cite{willems1972dissipative} that form the theoretical foundation for our analysis. We begin with classical passivity theory defined with respect to a fixed equilibrium point. We also consider the nonlinear negative imaginary (NNI) property, which generalizes the negative-imaginary concept from linear systems to nonlinear systems.

% \subsection{System Model and Equilibrium Points}

Consider the system $\Sigma$ as a general nonlinear dynamical system:
\begin{subequations}\label{eq:nonlinearmodel}
    \begin{align}
        \dot{x} &= f(x, u), \quad x \in \R^n,\; u \in \R^m, \label{eq:dyn}\\
        y &= g(x, u), \quad\;\; y \in \R^m, \label{eq:out}
    \end{align}
\end{subequations}
where \(f:\R^n\times\R^m\to\R^n\) and \(g:\R^n\times\R^m\to\R^m\) are locally Lipschitz functions. The input \(u(\cdot)\) is assumed to be piecewise continuous and bounded.

% The dissipativity properties discussed herein are defined with respect to the system's equilibrium points.

\begin{definition}
\label{def:equilibrium-set}
    The set of admissible equilibria for system \eqref{eq:nonlinearmodel} is defined as:
    \[
    \mathcal{E} := \left\{(x^\star, u^\star) \in \R^n \times \R^m \mid f(x^\star, u^\star) = 0 \right\}.
    \]
    For any equilibrium pair \((x^\star, u^\star) \in \mathcal{E}\), the corresponding equilibrium output is \(y^\star = g(x^\star, u^\star)\).
\end{definition}

Without loss of generality, the following assumption is made to simplify the analysis.
\begin{assumption}\label{ass:eq}
    There exists an equilibrium point \((x^\star, u^\star) \in \mathcal{E}\) such that the corresponding output vanishes; i.e., \(y^\star = g(x^\star, u^\star) = 0\). 
\end{assumption}

% Under this assumption, we define the shifted variables:
% \[
% \tilde{x} := x - x^\star, \quad
% \tilde{u} := u - u^\star, \quad
% \tilde{y} := y - y^\star = y.
% \]
% In these coordinates, the system equilibrium resides at the origin \((\tilde{x}, \tilde{u}, \tilde{y}) = (0, 0, 0)\).

% \textit{Passivity and Equilibrium-Independent Passivity}

The classical notion of passivity is defined with respect to a specific equilibrium point.

\begin{definition}\label{def:passivity}
    The system \eqref{eq:nonlinearmodel} is passive \cite{brogliato2007dissipative} about the equilibrium \((x^\star,u^\star,y^\star)=(0,0,0)\), if there exists a continuously differentiable storage function \(V:\R^n\to\R_{\geq 0}\) with \(V(0)=0\) such that
    \begin{equation}\label{eq:passivity_ineq}
        \dot{V}(x) \leq y^\top u.
    \end{equation}
    % \color{red}
    Furthermore, if the inequality
    \[
        \dot{V}(x) \leq y^\top u - \rho y^\top y.
    \]
    holds for some \(\rho>0\), the system is output strictly passive.
\end{definition}

% However, passivity about a single equilibrium can be restrictive when stability guarantees are required across multiple operating points. This limitation motivates the concept of equilibrium-independent passivity \cite{simpson2018equilibrium,monshizadeh2019conditions}.

% \begin{definition}\label{def:shifted_passivity}
%     System \eqref{eq:nonlinearmodel} is equilibrium-independent passive if for every equilibrium pair \((x^\star,u^\star)\in\mathcal{E}\), there exists a continuously differentiable storage function \(W_{x^\star}:\R^n\to\R_{\geq 0}\) with \(W_{x^\star}(x^\star)=0\) such that
%     \[
%         \dot{W}_{x^\star}(x) \leq (u - u^\star)^\top (y - y^\star).
%     \]
%     Moreover, if there exists a function \(\rho>0\) such that
%     \[
%         \dot{W}_{x^\star}(x) \leq (u - u^\star)^\top (y - y^\star) - \rho(\|y - y^\star\|)
%     \]
%     then the system is said to be equilibrium-independent strictly passive.
% \end{definition}

% \textit{Nonlinear Negative Imaginary Systems}

The negative imaginary (NI) property was first introduced for linear systems~\cite{petersen2010feedback} and has been used in the analysis of lightly damped mechanical and electrical networks. 
Both passivity and the NI properties characterize systems that do not generate energy. In the case of mechanical systems, the passivity property describes systems with force inputs and velocity outputs, whereas the NI property describes systems with force inputs and position outputs. To extend the NI property beyond linear systems, the nonlinear negative imaginary (NNI) property~\cite{ghallab2025negative} is defined.

\begin{definition}\label{def:nni}
    Suppose that \((x^\star, u^\star) = (0, 0)\) is an equilibrium of~\eqref{eq:nonlinearmodel} with \(y^\star = g(0,0) = 0\).  
    The system is nonlinear negative imaginary (NNI)  if there exists a continuously differentiable, nonnegative storage function \(V:\R^n \to \R_{\geq 0}\) such that, for all \(t \geq 0\)
    \begin{equation}\label{eq:nni_ineq}
        \dot{V}(x) \leq \dot{y}^\top u.
    \end{equation}
\end{definition}

This inequality is similar to the classical passivity condition, but with the supply rate \( \dot{y}^\top u \) instead of \( y^\top u \).  
Hence, passivity is related to the energy exchange between input and output, whereas the NNI property is related to the energy exchange between the input and the rate of change of the output.

% Analogous to the extension from passivity to equilibrium-independent passivity, we define the equilibrium-independent NNI property.

% \begin{definition}\label{def:shifted_nni}
%     System~\eqref{eq:nonlinearmodel} is equilibrium-independent NNI if for every equilibrium pair \((x^\star, u^\star) \in \mathcal{E}\), there exists a storage function \(W_{x^\star}:\R^n \to \R_{\geq 0}\) with \(W_{x^\star}(x^\star)=0\) such that
%     \begin{equation}\label{eq:shifted_nni_ineq}
%         \dot{W}_{x^\star}(x) \leq \dot{\tilde{y}}^\top \tilde{u},
%     \end{equation}
%     where \(\tilde{y} = y - y^\star\) and \(\tilde{u} = u - u^\star\).
%      \color{red} Moreover, if there exists a function \(\rho>0\) such that
%     \[
%         \dot{W}_{x^\star}(x) \leq (u - u^\star)^\top (y - y^\star) - \rho(\|y - y^\star\|)
%     \]
%     then the system is said to be equilibrium-independent strictly passive.
% \end{definition}

\section{Model of a Synchronous Machine Connected to a Current Source}\label{sec:modeling}

The mathematical modeling of synchronous machines (SMs) has been widely studied in the literature, with various levels of simplification depending on the intended analysis \cite{kundur2007power, grainger1999power, sauer2017power}. Most of the conventional models are typically derived under steady-state and balanced conditions to simplify the analysis.

In this study, following the electromechanical modeling approaches presented in \cite{zhong2010synchronverters}, we consider a simplified representation of a cylindrical-rotor synchronous generator with one pole pair per phase. The following assumptions are made to reduce model complexity, while capturing the essential dynamic behavior:

\begin{enumerate}
    \renewcommand{\labelenumi}{(\roman{enumi})}
\item The field current, $i_f$, is constant;
\item Damper windings, magnetic saturation, and eddy current effects are neglected;
\item The stator windings are balanced, sinusoidally distributed, and star-connected without a neutral connection.
\end{enumerate}

The considered system consists of a synchronous machine connected to an external network, as illustrated in Figure \ref{SG_T_R}. The stator terminals are linked to a load bus through a transmission line modeled by a series impedance with resistance $R_l$ and inductance $L_l$. At the load bus, a constant current source operates in parallel with a resistive load $R_L$.

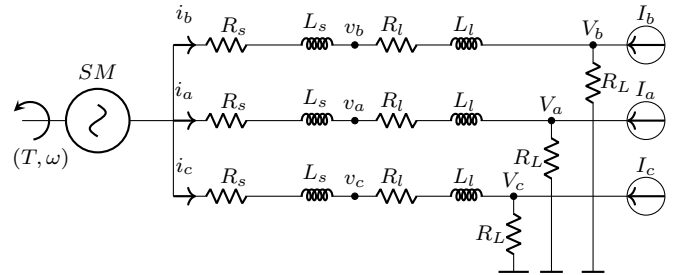
\begin{figure}[h]
\centering
\scalebox{1}{ % <-- scaling starts here

\begin{circuitikz}
    \small
    % Sinusoidal voltage source (Generator)
    \draw (-1,0) to[sinusoidal voltage source, l_=$ SM $] (-3,0);
    \draw[-stealth, black, thick] (-2.85,-.25) arc[start angle=-80,end angle=150,radius=0.25];
    \node at (-2.75,-0.5) {$(T,\omega)$};
    % \draw (0,-2) -- (5.1,-2);
    \draw [color=black](-1,0) to[R, bipoles/length=18pt] (.4,0) to[L, bipoles/length=18pt] (1.4,0)
        to[R, bipoles/length=18pt] (2.5,0) to[L, bipoles/length=18pt] (3.25,0)to (3.65,0)to (4.1,0) to (5,0); % smaller R and L
    \draw [color=black](-1,1) to[R, bipoles/length=18pt] (.4,1) to[L, bipoles/length=18pt] (1.4,1)
        to[R, bipoles/length=18pt] (2.5,1) to[L, bipoles/length=18pt] (3.25,1)to (3.65,1) to (4.55,1)to (5,1); % smaller R and L    
    \draw [color=black](-1,-1) to[R, bipoles/length=18pt] (.4,-1) to[L, bipoles/length=18pt] (1.4,-1)
        to[R, bipoles/length=18pt] (2.5,-1) to[L, bipoles/length=18pt] (3.25,-1) to (3.65,-1)to (5,-1); % smaller R and L    
    
    \draw [color=black] (3.5,-1) to[R, bipoles/length=18pt] (3.5,-2);
    \draw [color=black] (4,0) to[R, bipoles/length=18pt] (4,-1)to (4,-2);
    \draw [color=black] (4.55,1) to[R, bipoles/length=18pt] (4.55,0) to (4.55,-2);
    %%%%%%%%%%%%%%%%%%%%%%%%%%%%%%%%%%%%%%%%%%%%%%%%%%%%%%%%%%%%%%%
    \draw[fill=black] (1.4,1) circle (1.2pt);
    \draw[fill=black] (1.4,0) circle (1.2pt);
    \draw[fill=black] (1.4,-1) circle (1.2pt);
    %%%%%%%%%%%%%%%%%%%%%%%%%%%%%%%%%
    \draw[fill=black] (4.55,1) circle (1.2pt);
    \draw[fill=black] (4,0) circle (1.2pt);
    \draw[fill=black] (3.5,-1) circle (1.2pt);
    %%%%%%%%%%%%%%%%%% I %%%%%%%%%%%%%%%%%%%%
    \draw (5.25,1) circle (.25cm);
    \draw (5.25,0) circle (.25cm);
    \draw (5.25,-1) circle (.25cm);
    \draw[->,thick] (5.5,1) -- (5,1);
    \draw[->,thick] (5.5,0) -- (5,0);
    \draw[->,thick] (5.5,-1) -- (5,-1);
    %%%%%%%%%%%% I %%%%%%%%%%%%%%%
    \color{black}
    \node at  (5.25,1.4)   {$I_b$};
    \node at  (5.25,.4)   {$I_a$};
    \node at  (5.25,-.6)   {$I_c$};                
    \color{black}    
    %%%%%%%%%%%%%%%%%%%%%%%%%%%%%%%%%
    \draw (-1,1) -- (-1,-1);    
    %%%%%%%%%%%%%%%%%%%%%%%%%%%%%%%%%
    \draw[thick] (4.4,-2) -- (4.7,-2);
    \draw[thick] (3.85,-2) -- (4.15,-2);
    \draw[thick](3.3,-2) -- (3.7,-2); 
    %%%%%%%%%%%  R_L %%%%%%%%%%%%%%%%
    \node at (3.2,-1.5) {$R_L$};
    \node at (3.7,-.5) {$R_L$};
    \node at (4.85,0.5) {$R_L$};
    %%%%%%%%%%%  R_s %%%%%%%%%%%%%%%%
    \node at (-.2,1.25) {$R_s$};
    \node at (-.2,.25) {$R_s$};
    \node at (-.2,-.75) {$R_s$};    
    %%%%%%%%%%%  L_s %%%%%%%%%%%%%%%%
    \node at (.88,1.3) {$L_s$};
    \node at (.88,.3) {$L_s$};
    \node at (.88,-.7) {$L_s$};  
    %%%%%%%%%%%  R %%%%%%%%%%%%%%%%
    \node at (1.9,1.25) {$R_l$};
    \node at (1.9,.25) {$R_l$};
    \node at (1.9,-.75) {$R_l$};    
    %%%%%%%%%%%  L %%%%%%%%%%%%%%%%
    \node at (2.85,1.25) {$L_l$};
    \node at (2.85,.25) {$L_l$};
    \node at (2.85,-.75) {$L_l$};    
    %%%%%%%%%%%%V_abc%%%%%%%%%%%%%%%
    \color{black}
    \node at  (1.4,1.2)   {$v_b$};
    \node at (1.4,.2)   {$v_a$};
    \node at  (1.4,-.8)   {$v_c$};     
    %%%%%%%%%%%% V %%%%%%%%%%%%%%%
    \color{black}
    \node at  (4.55,1.2)   {$V_b$};
    \node at (4,.2)   {$V_a$};
    \node at  (3.5,-.8)   {$V_c$};        
    \draw[->, thick] (-1,0) -- (-.7,0) node[midway, above, yshift=5pt]{$i_{a}$};
    \draw[->, thick] (-1,1) -- (-.7,1) node[midway, above, yshift=5pt]{$i_{b}$};
    \draw[->, thick] (-1,-1) -- (-.7,-1) node[midway, above, yshift=5pt]{$i_{c}$};         
\end{circuitikz}}
\caption{SM connected to an impedance load via transmission line $L_l$, a current source $I_{abc}$ representing the external network.}
\label{SG_T_R}
\end{figure} 
To model the electrical dynamics of the SM, $\theta$ and $\omega = \dot{\theta}$ denote the rotor angle and angular velocity, respectively. The stator winding resistance and inductance are represented by $R_s$ and $L_s$. The rotor carries a field winding that produces a magnetic field characterized by the field current $i_f$. The magnetic coupling between the constant rotor field and the stator windings is represented by the mutual inductance $M_f$. Three-phase stator quantities, such as the current $i_{abc} = [i_a, i_b, i_c ]^\top$ and the induced voltage $e_{abc} = [ e_a, e_b, e_c ]^\top$, are represented as vectors in the $abc$ reference frame. Consequently, the back electromotive force (EMF) induced in the stator windings is expressed as
\[
e_{abc} = M_f i_f \omega
\begin{bmatrix}
\sin\theta\\
\sin(\theta-\tfrac{2\pi}{3})\\
\sin(\theta-\tfrac{4\pi}{3})
\end{bmatrix}.
\]

The external current source is assumed to operate at a fixed angular velocity $\omega^s$. Choosing the phase reference as zero gives $\theta^s = \omega^s t$, leading to the following current vector
\[
I_{abc} = I_s
\begin{bmatrix}
    \sin\theta^s \\
    \sin\left(\theta^s - \frac{2\pi}{3}\right) \\
    \sin\left(\theta^s - \frac{4\pi}{3}\right)
\end{bmatrix}.
\]
where $I_s$ is the current amplitude, and $I_{abc} = [ I_a, I_b, I_c ]^\top$ denotes the three-phase current source vector. Applying Kirchhoff's current law, yields the stator current dynamics
\[
L \frac{di_{abc}}{dt} = -(R_s + R_l) i_{abc} + e_{abc} - R_L(i_{abc} + I_{abc}).
\]

where $L = L_s + L_l$ represents the total inductance of the stator and transmission line. The mechanical dynamics of the rotor follow Newton’s second law
\[
J\dot{\omega} = -D\omega + T_m - T_e,
\]
where $J$ is the rotor moment of inertia, $D$ is the damping coefficient, $T_m$ is the mechanical input torque, and $T_e$ is the electromagnetic torque. For a constant field current, the electromagnetic torque can be expressed as \cite{zhong2010synchronverters}
\[
T_e = \frac{1}{\omega} e_{abc}^\top i_{abc}.
\]
By combining the electrical and mechanical subsystems, the complete system model in the $abc$-frame is
\begin{equation}
\begin{aligned}
    \dot{\theta} &= \omega, \\
    J \dot{\omega} &= -D\omega - \frac{1}{\omega} e_{abc}^\top i_{abc} + T_m, \\
    L \frac{d i_{abc}}{dt}  &= -R i_{abc}  + e_{abc} - R_L I_{abc},
\end{aligned}
\label{eq:fullmodel_abc}
\end{equation}
where $R = R_s + R_l + R_L$ represents the total series resistance of the stator, transmission line, and load. The model  (\ref{eq:fullmodel_abc}) captures the essential electromechanical interaction between the SM and the external network, which includes a current source at the load bus.

% \textit{DQ-Transformation}
For the analysis of the three-phase model derived in (\ref{eq:fullmodel_abc}), we describe all of the electrical quantities in the $dq$-reference frame. 
This transformation converts the time-varying sinusoidal signals into constant quantities, thereby simplifying the analysis and interpretation of the system dynamics. The standard Park transformation \(T_{dq}(\varrho)\) is defined as
\[
T_{dq}(\varrho) = \sqrt{\frac{2}{3}} 
\begin{bmatrix}
    \cos(\varrho) & \cos\left(\varrho - \frac{2\pi}{3}\right) & \cos\left(\varrho - \frac{4\pi}{3}\right) \\
    \sin(\varrho) & \sin\left(\varrho - \frac{2\pi}{3}\right) & \sin\left(\varrho - \frac{4\pi}{3}\right)
\end{bmatrix},
\]
where \(\varrho\) is the transformation angle. Generally, any rotating reference angle can be used. However, in this study, we choose as the instantaneous rotor angle:
\[
    \varrho := \theta(t).
\]

By applying the $dq$ transformation, the reference frame is aligned with the rotor magnetic field, such that the $d$-axis coincides with the field axis. 
The derivative of the current vector in the $dq$ coordinate satisfies
\[
\begin{aligned}
\frac{di_{dq}}{dt} &= \frac{dT_{dq}(\varrho)}{dt} i_{abc} + T_{dq}(\varrho) \frac{di_{abc}}{dt} \\
&= \dot{\varrho} 
\begin{bmatrix}
-i_q \\[4pt]
i_d
\end{bmatrix}
+ T_{dq}(\varrho) \frac{di_{abc}}{dt} 
= \omega
\begin{bmatrix}
-i_q \\[4pt]
i_d
\end{bmatrix}
+ T_{dq}(\varrho)\frac{di_{abc}}{dt}.
\end{aligned}
\]

Similarly, the external current source is expressed in the dq-reference frame as follows, where the magnitude is redefined as $I =\sqrt{\frac{3}{2}} I_s$:

\[
I_{dq} =
\begin{bmatrix}
I_{d} \\[4pt]
I_{q}
\end{bmatrix}
= T_{dq}(\varrho) I_{abc}
= I
\begin{bmatrix}
\sin(\theta^s - \theta) \\[4pt]
\cos(\theta^s - \theta)
\end{bmatrix}.
\]
and the voltage at the load bus is defined as 
\[
V_{dq} =
\begin{bmatrix}
V_{d} \\[4pt]
V_{q}
\end{bmatrix}
= T_{dq}(\varrho) V_{abc}
= R_LI
\begin{bmatrix}
\sin(\theta^s - \theta) \\[4pt]
\cos(\theta^s - \theta)
\end{bmatrix}.
\]
By substitution, the stator current dynamics become
\[
L\frac{di_{dq}}{dt}
= -Ri_{dq}
+ L\omega
\begin{bmatrix}
-i_q \\[4pt]
i_d
\end{bmatrix}
- R_L I_{dq}.
\]

Defining the constant \(b = M_f i_f \sqrt{3/2}\), the complete electromechanical model in the $dq$-coordinates is
\begin{equation}
\begin{aligned}
    \dot{\delta} &= \omega - \omega^s, \\
    J \dot{\omega} &= -D \omega - b i_q + T_m, \\
    L \dot{i}_{d} &= - R i_d + L \omega i_q - R_L I \sin\delta, \\
    L \dot{i}_{q} &= - R i_q - L \omega i_d + b \omega - R_L I \cos\delta.
\end{aligned}
\label{eq:full_model_dq}
\end{equation}

where \(\delta = \theta^s - \theta\) represents the rotor angle deviation with respect to the angle of the external current source.
\begin{remark}
The classical swing equation, which describes the rotor dynamics of a synchronous machine connected to an infinite bus, can be derived from the above electromechanical model using suitable approximations. 
To derive this model, the electrical subsystem is assumed to evolve much faster than the mechanical subsystem. Consequently, the current dynamics are assumed to reach their steady-state values almost instantaneously relative to the slower mechanical variables. 
Thus, we set the derivatives \(\dot{i}_d, \dot{i}_q\) equal to zero, implying that the fast electrical subsystem is replaced by its quasi–steady-state manifold, consistent with the approximations in \cite{sauer2017power}. Furthermore, by neglecting the small resistance of the stator and transmission line; i.e., \(R_l+R_s \approx 0\) and considering an infinite bus, the electrical subsystem reduces to
\[
\begin{aligned}
    0 &= L \omega i_q - V \sin\delta, \\
    0 &= -L \omega i_d + b \omega - V \cos\delta.
\end{aligned}
\]
Substituting the corresponding expression for \(i_q\) into the mechanical equation yields the improved swing equation in \cite{zhou2008improved}:
\[
\begin{aligned}
    J \dot{\omega} &= -D \omega - \frac{b V \sin\delta}{L \omega} + T_m.
\end{aligned}
\]
Deriving the classical swing equation requires two additional approximations.
\begin{enumerate}
    \item The rotor speed \(\omega\) in the denominator of the electrical torque term is approximated by the constant synchronous speed \(\omega^s\). This approximation is valid for small frequency deviations around the operating point.
    \item The equation is multiplied through by \(\omega^s\) to convert it from a torque balance equation to a power balance equation using the relation \(P = T \omega^s\).
\end{enumerate}

Performing these steps, and defining the new damping coefficient \(\tilde{D} = D \omega^s\) and the mechanical power \(P_m = T_m \omega^s\), we arrive at the classical form of the swing equation

where \(P_{\text{max}} = \frac{M_f i_f V}{\omega_s L_s}\) is the maximum synchronizing power. This demonstrates that, under the standard assumptions of fast stator dynamics, negligible resistance, and small speed deviations, the proposed electromechanical model reduces to the classical swing equation.
\end{remark}

\section{Main results}\label{sec:main}

The following section considers the inherent robustness of the electromechanical model \eqref{eq:full_model_dq} using dissipativity theory. We prove that the system shows several energy-based properties under different input–output configurations, including passivity, passivity of the shifted system, nonlinear negative-imaginary behavior, and the NNI property of the shifted system. 
% All results are based on the derived electromechanical model, with appropriate input–output selections and equilibrium shifts revealing distinct energy-based properties.
Furthermore, we show that both passivity and NNI properties are preserved under interconnection with passive or negative-imaginary controllers, such as various droop-type controllers.

% \textit{Equilibrium Point}

The equilibrium state of system~\eqref{eq:full_model_dq} is denoted by 
\((\delta^s, \omega^s, i_d^s, i_q^s)\), corresponding  to the steady-state 
inputs represented by \((I_d^s, I_q^s)\). 
At the equilibrium, all time derivatives vanish:
\begin{equation}
    \dot{\omega} = 0, \quad
    \dot{i}_d = 0, \quad
    \dot{i}_q = 0.
    \label{eq:equilibrium_condition}
\end{equation}
Substituting the conditions in (\ref{eq:equilibrium_condition}) into the system equations yields the steady-state relations
\begin{equation}
\begin{aligned}
    \omega &= \omega^s, \\
    0 &= -D \omega^s - b i_q^s + T_m, \\
    0 &= -R i_d^s + L \omega^s i_q^s - R_L I \sin\delta^s, \\
    0 &= -R i_q^s - L \omega^s i_d^s + b \omega^s - R_L I \cos\delta^s.
\end{aligned}
\label{eq:equilibrium}
\end{equation}

\begin{theorem}
Consider the system~\eqref{eq:full_model_dq} with state variables 
\((\delta,\omega, i_d, i_q)\) and input current source \(I_{dq}\), where the mechanical torque input is set to zero (\(T_m = 0\)).  
Let the output be the load-bus voltage \((V_{dq})\), and define
\[\begin{aligned}
    u_1 &:= I
    \begin{bmatrix}
    \sin(\theta^s - \theta) \\[4pt]
    \cos(\theta^s - \theta)
    \end{bmatrix}=\begin{bmatrix}
        I_d\\
        I_q
    \end{bmatrix}, \\
    y_1 &:= R_L
    \left(
    \begin{bmatrix}
    i_d \\[4pt]
    i_q
    \end{bmatrix}
    + I
    \begin{bmatrix}
    \sin(\theta^s - \theta) \\[4pt]
    \cos(\theta^s - \theta)
    \end{bmatrix}
    \right)=\begin{bmatrix}
        V_d\\
        V_q
    \end{bmatrix}.
\end{aligned}\]
Then, the system is passive with respect to the input-output pair \((u_1, y_1)\).
\end{theorem}

\begin{proof}
Let the storage function represent the total stored energy of the system, consisting of the rotor’s kinetic energy and the magnetic energy stored in the inductors:
\begin{equation}
V(x) = \frac{1}{2} J \omega^2 + \frac{1}{2} L i_d^2 + \frac{1}{2} L i_q^2.
\label{eq:Storagefun}
\end{equation}
Taking the time derivative of \(V(x)\) along the system trajectories~\eqref{eq:full_model_dq} gives
\begin{equation}
\begin{aligned}
\dot{V}(x) = {} & -D\omega^2 - b i_q \omega 
- R i_d^2 + L \omega i_q i_d - R_L I i_d\sin\delta  \\
& - R i_q^2 - L \omega i_d i_q + b \omega i_q 
- R_L I i_q\cos\delta .
\end{aligned}
\label{eq:derivativeV}
\end{equation}

The supplied power is given by
\[
\begin{aligned}
y^\top_1 u_1 =V_{dq}^\top I_{dq} = {} & R_L i_d I \sin\delta + R_L I^2 \sin^2\delta \\
& + R_L i_q I \cos\delta + R_L I^2 \cos^2\delta.
\end{aligned}
\]
The dissipation inequality becomes
\[
\begin{aligned}
\dot{V}(x) - y^\top_1 u_1 &= -D \omega^2 - R i_d^2 - R i_q^2 
- 2 R_L i_d I \sin\delta\\
& - 2 R_L i_q I \cos\delta 
 - R_L I^2 (\sin^2\delta + \cos^2\delta).
\end{aligned}
\]

Since \(R = R_s + R_l + R_L > 0\), the expression is non-positive after completing the squares for the terms in \(i_d\) and \(i_q\). This leads to the inequality
\[
\dot{V}(x) \leq y_1^\top u_1,
\]
which establishes the passivity condition. That is, the system is passive with respect to the input-output pair \((u_1, y_1)\).
\end{proof}
\medskip
\noindent
Next, we demonstrate the same physical system exhibits the nonlinear negative-imaginary property when viewed from a torque input and rotor angle output.

\begin{theorem}
Consider the system~\eqref{eq:full_model_dq} with state variables 
\((\delta, \omega, i_d, i_q)\), where the input is the electromagnetic torque and the output is the rotor angle:
\[
u_2 = T_e = \frac{V_d I_d + V_q I_q}{\omega} , \qquad y_2 = \theta.
\]
Then, the system is nonlinear negative imaginary with respect to the input-output pair \((u_2, y_2)\).
\end{theorem}

\begin{proof}
Using the same energy-based storage function as in~\eqref{eq:Storagefun},
the NNI property requires that
\[
\dot{V}(x) \leq \dot{y}^\top_2 u_2.
\]
Since \(\dot{y}_2 = \dot{\theta} = \omega\), this condition simplifies to
\begin{equation}
    \dot{V}(x) \leq  V_d I_d + V_q I_q.
    \label{eq:NNIcondition}
\end{equation}

The time derivative of \(V(x)\) along the system trajectories is given by~\eqref{eq:derivativeV}. Substituting \eqref{eq:derivativeV} and rearranging yields
\[
\begin{aligned}
\dot{V}(x) - (V_d I_d + V_q I_q)
= {} & -D \omega^2 - (R_s + R_l)(i_d^2 + i_q^2) \\
& - R_L (i_d + I \sin(\delta))^2 \\
& - R_L (i_q + I \cos(\delta))^2.
\end{aligned}
\]
The right-hand side is non-positive because all coefficients are positive. Therefore, the inequality~\eqref{eq:NNIcondition} holds for all trajectories of the system. Hence, the system satisfies the NNI property with respect to the input-output pair \((u_2, y_2)\).
\end{proof}

To analyze the system behavior in a neighborhood of the equilibrium point, 
we introduce the shifted variables, defined as the deviations from the equilibrium point  corresponding to~\eqref{eq:full_model_dq},
\[
 \delta:= \tilde\delta+\delta^s,\quad
 \omega:= \tilde\omega+\omega^s,\quad
 i_d:= \tilde i_d+i_d^s,\quad
 i_q:= \tilde i_q+i_q^s.
\]
Substituting the shifted variables into the electromechanical modeling of the system \eqref{eq:full_model_dq} yields 
\[
    \begin{aligned}
    \dot{\tilde\delta} &= \tilde\omega, \\[4pt]
    \dot{\tilde\omega} &= -D (\tilde\omega+\omega^s) - b(\tilde i_q+i_q^s), \\[6pt]
    \dot{\tilde i}_d &= -R(\tilde i_d+i_d^s) 
        - L(\tilde\omega+\omega^s)(\tilde i_q+i_q^s)  - R_L I\sin(\tilde\delta+\delta^s), \\[6pt]
    \dot{\tilde i}_q &= -R(\tilde i_q+i_q^s) 
        + L(\tilde\omega+\omega^s)+ b (\tilde\omega+\omega^s) \\
         & \qquad - R_L I\cos(\tilde\delta+\delta^s).
    \end{aligned}
\]

using the equilibrium conditions \eqref{eq:equilibrium}, the resulting perturbed system dynamics are
\begin{equation}
    \begin{aligned}
    \dot{\tilde\delta} &= \tilde\omega, \\[4pt]
    \dot{\tilde\omega} &= -D\,\tilde\omega - b\,\tilde i_q, \\[6pt]
    \dot{\tilde i}_d &= -R\,\tilde i_d 
        - L\tilde\omega\,\tilde i_q 
        - L\tilde\omega\, i_q^s 
        - L\omega^s \tilde i_q \\
        &\quad - R_L I\big(\sin(\tilde\delta+\delta^s)-\sin\delta^s\big), \\[6pt]
    \dot{\tilde i}_q &= -R\,\tilde i_q 
        + L\tilde\omega\,\tilde i_d 
        + L\tilde\omega\, i_d^s 
        + L\omega^s \tilde i_d + b \tilde\omega \\
        &\quad - R_L I\big(\cos(\tilde\delta+\delta^s)-\cos\delta^s\big).
    \end{aligned}
    \label{eq:shifted_system}
\end{equation}

\begin{theorem}
Consider the shifted system~\eqref{eq:shifted_system} with the state variables 
\(\tilde{x} = [\,\tilde{\delta},\, \tilde{\omega},\, \tilde{i}_d,\, \tilde{i}_q\,]^\top\).
Define the corresponding shifted input and output as
\begin{equation}
    \begin{aligned}
    \tilde{u}_1 &=
    \begin{bmatrix}
    I \big(\sin(\tilde{\delta} + \delta^s) - \sin\delta^s\big) \\[4pt]
    I \big(\cos(\tilde{\delta} + \delta^s) - \cos\delta^s\big)
    \end{bmatrix}, \\
    \tilde{y}_1 &=
    R_L
    \begin{bmatrix}
    I \big(\sin(\tilde{\delta} + \delta^s) - \sin\delta^s\big) + \tilde{i}_d \\[4pt]
    I \big(\cos(\tilde{\delta} + \delta^s) - \cos\delta^s\big) + \tilde{i}_q
    \end{bmatrix}.
    \end{aligned}
\label{pass_shifted_io}
\end{equation}

Then, the shifted system is passive from the input \(\tilde{u}_1\) to the output \(\tilde{y}_1\) if the condition
\begin{equation}
    (i_d^s)^2 + (i_q^s)^2 \leq \frac{4D(R_s+R_l)}{L^2}.      
    \label{dissipativity condition}
\end{equation}

is satisfied. Moreover, the system is strictly passive when the inequality in  \eqref{dissipativity condition} is strict.

\end{theorem}
\color{black}
\begin{proof}
Consider the storage function  representing the energy of the shifted system as
\begin{equation}
V(\tilde{x}) = \frac{1}{2}J\tilde{\omega}^2
+ \frac{1}{2}L\tilde{i}_d^2
+ \frac{1}{2}L\tilde{i}_q^2.
\label{eq:shifted_storage}
\end{equation}
The trigonometric differences are defined as follows
\[
\cos\Delta = \cos(\tilde{\delta} + \delta^s) - \cos\delta^s, 
\quad
\sin\Delta = \sin(\tilde{\delta} + \delta^s) - \sin\delta^s.
\]
Differentiating~\eqref{eq:shifted_storage} along the trajectories of the shifted system yields
\[
\begin{aligned}
\dot{V} = {} & -D \tilde{\omega}^2 - R (\tilde{i}_d^2 + \tilde{i}_q^2)
+ L\tilde{\omega}(\tilde{i}_q i_d^s - \tilde{i}_d i_q^s)\\
& - R_L I\big(\tilde{i}_d \sin\Delta + \tilde{i}_q \cos\Delta\big),
\end{aligned}
\]

and the supplied power is given as
\[
\begin{aligned}
\tilde{u}^\top_1 \tilde{y}_1
=& R_L \big[ I\sin\Delta\, (I\sin\Delta + \tilde{i}_d)
+ I\cos\Delta\, (I\cos\Delta + \tilde{i}_q) \big] \\
= R_L I^2 & \big[ (\sin\Delta)^2 + (\cos\Delta)^2 \big]
+ R_L I \big( \tilde{i}_d \sin\Delta + \tilde{i}_q \cos\Delta \big).
\end{aligned}
\]

Subtracting the supplied power from the storage rate gives
\[
\begin{aligned}
\dot{V} - \tilde{u}^\top_1 \tilde{y}_1 = {} &
- R_L \big[ (\tilde{i}_d+I\sin\Delta)^2 + (\tilde{i}_q+I\cos\Delta)^2 \big]
- D \tilde{\omega}^2  \\
& - (R_s+R_l) (\tilde{i}_d^2 + \tilde{i}_q^2) - L \tilde{\omega} (-\tilde{i}_d i_q^s + \tilde{i}_q i_d^s).
\end{aligned}
\]

The first term, \(- R_L \big[ (\tilde{i}_d+I\sin\Delta)^2 + (\tilde{i}_q+I\cos\Delta)^2 \big]\), is strictly negative, and the remaining terms can be expressed in the  following quadratic form where \(R_s+R_l=R_{sl}\):
\[
\begin{aligned}
&- D \tilde{\omega}^2 - R_{sl}(\tilde{i}_d^2 + \tilde{i}_q^2)
- L \tilde{\omega} (-\tilde{i}_d i_q^s + \tilde{i}_q i_d^s)
\\[4pt]
&= -\begin{bmatrix}
\tilde{i}_d \\[2pt] \tilde{i}_q \\[2pt] \tilde{\omega}
\end{bmatrix}^{\!\top}
\begin{bmatrix}
R_{sl} & 0 & \tfrac{L i_q^s}{2} \\[3pt]
0 & R_{sl} & -\tfrac{L i_d^s}{2} \\[3pt]
\tfrac{L i_q^s}{2} & -\tfrac{L i_d^s}{2} & D
\end{bmatrix}
\begin{bmatrix}
\tilde{i}_d \\[2pt] \tilde{i}_q \\[2pt] \tilde{\omega}
\end{bmatrix}.
\end{aligned}
\]

Applying the Schur complement, a sufficient condition for this quadratic form to be negative semidefinite is
\begin{equation}
D \geq \frac{L^2 \big( (i_d^s)^2 + (i_q^s)^2 \big)}{4R_{sl}}.
\label{eq:strict-cond}
\end{equation}
When condition~\eqref{eq:strict-cond} holds, the dissipation inequality 
\(\dot{V} \leq \tilde{u}^\top_1 \tilde{y}_1\) is satisfied, implying passivity of the shifted system.  
If the inequality \eqref{eq:strict-cond} is strict, the system is strictly passive. 
\end{proof}
Building on the shifted model \eqref{eq:shifted_system}, we now examine its nonlinear negative-imaginary property. 
The output is defined as the shifted rotor angle from the equilibrium, and the input is taken as the electromagnetic torque, expressed as the ratio of the electrical power to the angular velocity:
\begin{equation}
\begin{aligned}
\tilde{u}_2 &:= \frac{R_L\big[ \tilde{I}_d(\tilde{I}_d-\tilde{i}_d)+\tilde{I}_q(\tilde{I}_q-\tilde{i}_q)\big]}{\tilde{\omega}}=T_e-T_e^s,\\
\tilde{y}_2 &:= \tilde{\delta}.
\end{aligned}
\end{equation}
The supply rate is therefore given by
\begin{equation}
\dot{\tilde{y}}_2\,\tilde{u}_2
= R_L\big[ \tilde{I}_d(\tilde{I}_d-\tilde{i}_d)+\tilde{I}_q(\tilde{I}_q-\tilde{i}_d)\big].
\end{equation}

\begin{theorem}
\label{thm:shifted-nni}
Consider the shifted system~\eqref{eq:shifted_system} with the input–output pair 
\((\tilde{u}_2, \tilde{y}_2)\).  
The system is nonlinear negative-imaginary (NNI) if the condition (\ref{dissipativity condition}) holds.
% \begin{equation}
%     \frac{L^2 \big( (i_d^s)^2 + (i_q^s)^2 \big)}{4R} \le D
% \label{eq:shifted-nni-cond}
% \end{equation}
Moreover, if the inequality in~\eqref{dissipativity condition} is strict, the system is strictly nonlinear negative-imaginary.
\end{theorem}

\begin{proof}
The proof follows directly from the storage function~\eqref{eq:shifted_storage} 
and the same energy-balance argument employed in proving the passivity of the shifted system.  
When the quadratic form associated with the dissipation term is negative semidefinite, the inequality
\[
\dot{V} \le \dot{\tilde{y}}_2\,\tilde{u}_2
\]
holds, thereby establishing the  NNI property for the shifted system.  
If the condition in~\eqref{dissipativity condition} is strict, the quadratic form becomes negative definite, 
implying the strict NNI property.
\end{proof}

In what follows, we show that the passivity of the shifted system is preserved when the SM is interconnected with a droop-based torque controller. Importantly, the same supply rate is retained, which shows that the controller maintains the system’s inherent passivity.

Figure \ref{feedbackInterconnection} illustrates the feedback interconnection between the plant $(P)$ that representing the SM‌ and a torque controller $(C)$ at the mechanical port.

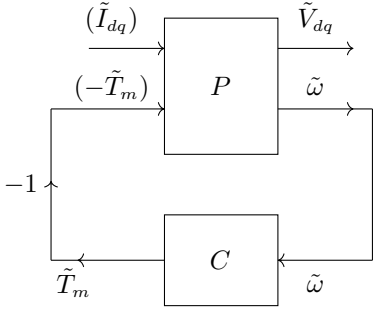
\begin{figure}[htbp]
    \centering
    \begin{minipage}{0.5\textwidth}
        \centering
        \begin{circuitikz}
            % Plant box
            \draw (3,-.6) rectangle (4.5,1.2);
            \node at (3.74,.3) {$P$};

            % Controller box
            \draw (3,-2.6) rectangle (4.5,-1.4);
            \node at (3.74,-2) {$C$};

            % Horizontal connections
            \draw (4.5,0) -- (5.75,0);
            \draw (1.5,0) -- (3,0);
            \draw (4.5,-2) -- (5.75,-2);
            \draw (1.5,-2) -- (3,-2);

            % Vertical interconnection
            \draw (1.5,0) -- (1.5,-2);
            \draw (5.75,0) -- (5.75,-2);
            
            % Labels
            \node at (2.3,1.2) {$(\tilde{I}_{dq})$};
            \node at (5,1.2) {$\tilde{V}_{dq}$};
            
            \node at (2.3,.3) {$(- \tilde{T}_m)$};
            \node at (1.8,-2.3) {$\tilde{T}_m$};
            \node at (5,.3) {$\tilde{\omega}$};
            \node at (5,-2.3) {$\tilde{\omega}$};

            % Arrows
            \draw[->] (4.55,-2) -- (4.5,-2);
            \draw[->] (2,.8) -- (3,.8);
            \draw[->] (2.95,0) -- (3,0);
            \draw[->] (4.5,.8) -- (5.5,.8);
            \draw[->] (5.485,0) -- (5.5,0);
            \draw[->] (1.9,-2) -- (1.85,-2);

            \draw[->] (1.5,-1) -- (1.5,-.951);
            \node at (1.1,-1) {$-1$};
            % \draw (1.5,-1) circle (10pt);
        \end{circuitikz}
        \caption{Feedback interconnection of the SM ($P$) and a torque controller ($C$). The plant includes both an electrical port $(\tilde u_e,\tilde y_e)$ and a mechanical port $(\tilde u_m,\tilde y_m)$}
        \label{feedbackInterconnection}
    \end{minipage}
\end{figure}

The electromechanical modeling of SM  is interpreted as a two-port system:  
an electrical port with the shifted pair \((\tilde u_e,\tilde y_e)\), as defined in (\ref{pass_shifted_io}), and a mechanical port defined by
\[
u_m := \tilde T_m=T_m-T_m^s, 
\qquad 
y_m := \tilde\omega.
\]

For the shifted model~\eqref{eq:shifted_system}, the storage function \( V_p \) is defined similarly to (\ref{eq:shifted_storage}). As shown earlier, the shifted system is passive under the condition (\ref{eq:strict-cond}) for the input-output pair \((\tilde u_e,\tilde y_e)\).

Next, consider a torque controller that is passive, such as a droop-based PI controller, defined by 

\begin{equation}
\tilde T_m = k_p \tilde\omega + k_i \tilde z, 
\qquad 
\dot{\tilde z} = \tilde\omega, 
\qquad k_p,k_i \ge 0.
\label{controller}
\end{equation}

This PI controller admits the storage function \( V_c = \tfrac12 k_i \tilde z^2 \), for which
\[
\dot V_c 
= k_i \tilde z \dot{\tilde z}
= (\tilde T_m - k_p \tilde\omega)\tilde\omega
\;\le\; \tilde{T_m} \tilde{\omega}=u_m y_m .
\]
Hence, the controller is passive from the input \( u_m = \tilde\omega \) to the output \( y_m = \tilde T_m \).

\vspace{.3cm}

\begin{theorem}
\label{thm:shifteddroop-passive}
Consider the shifted system~\eqref{eq:shifted_system} interconnected with the passive controller~\eqref{controller} via negative feedback, with total storage function $V_t = V_p + V_c$. Then, the closed-loop system is passive with respect to the original input-output pair defined in~\eqref{pass_shifted_io}, provided the condition in~\eqref{eq:strict-cond} holds. Furthermore, if the condition \eqref{eq:strict-cond} is strict, the closed-loop system is strictly passive.
\end{theorem}

\begin{proof}
Consider the storage function as the summation of the energy of the shifted system and the controller 
\begin{equation}
V_t(\tilde{x}) = \frac{1}{2}J\tilde{\omega}^2
+ \frac{1}{2}L\tilde{i}_d^2
+ \frac{1}{2}L\tilde{i}_q^2+\frac{1}{2}k_i\tilde{z}^2.
\label{eq:shifted_storagewithcont}
\end{equation}
Differentiating~\eqref{eq:shifted_storagewithcont} along the trajectories of the closed loop system results in
\[
\begin{aligned}
\dot{V}_t(x) = {} & -D \tilde{\omega}^2 - R (\tilde{i}_d^2 + \tilde{i}_q^2)
+ L\tilde{\omega}(\tilde{i}_q i_d^s - \tilde{i}_d i_q^s)\\
& - R I\big(\tilde{i}_d \sin\Delta + \tilde{i}_q \cos\Delta\big) -\tilde{T}_m\tilde{\omega}\\
& +(\tilde T_m - k_p \tilde\omega)\tilde\omega,
\end{aligned}
\]

and the supplied power is not changed
\[
\begin{aligned}
\tilde{u}^\top_e \tilde{y}_e
&= R I^2 \big[ (\sin\Delta)^2 + (\cos\Delta)^2 \big]
- R I \big( \tilde{i}_d \sin\Delta + \tilde{i}_q \cos\Delta \big).
\end{aligned}
\]

Subtracting the supplied power from the storage rate gives
\[
\begin{aligned}
\dot{V} - \tilde{u}^\top_e \tilde{y}_e = {} &
- R I^2 \big[ (\sin\Delta)^2 + (\cos\Delta)^2 \big]
- D \tilde{\omega}^2  \\
& - R (\tilde{i}_d^2 + \tilde{i}_q^2) - L \tilde{\omega} (-\tilde{i}_d i_q^s + \tilde{i}_q i_d^s)-k_p \tilde\omega^2.
\end{aligned}
\]

Again, we have a similar condition as in (\ref{eq:strict-cond}) for the preservation of the passivity property.
\end{proof}

\section{Conclusion and Future Work}\label{sec:conclusion}
This paper presented a nonlinear model of a synchronous machine connected to a transmission line, a resistive load, and a
current source representing the remainder of a network. This model derived from three phase $abc$ equations and then expressed in a rotor-aligned $dq$ reference frame. Using the electromechanical model, we established passivity from the current input to a voltage output. We also prove the passivity of a shifted system under explicit conditions on the parameters and steady-state values $(D, R, L, i_{dq}^s)$. In addition, we demonstrate a nonlinear negative imaginary property from an electrical torque input to rotor angle output. We further showed that incorporating a passive droop controller at the mechanical port preserves the passivity property. Future work will extend the analysis to include voltage and field dynamics.

% This paper establishes fundamental dissipativity properties of synchronous machines:

%  Proven incremental passivity from current to voltage
%  Demonstrated NNI property from power to angle
%  Unified swing equation with full dynamics
%  Constructive Lyapunov certificates for stability

% These results provide theoretical foundation for:
% \begin{itemize}
%     \item Stability assessment of renewable-rich grids
%     \item Design of grid-forming converters
%     \item Transient stability enhancement
% \end{itemize}
% Future work extends to: 1) Heterogeneous grids 2) Voltage dynamics 3) Experimental validation.

\bibliography{ifacconf}             % bib file to produce the bibliography

@book{machowski2020power,
  title={Power system dynamics: stability and control},
  author={Machowski, Jan and Lubosny, Zbigniew and Bialek, Janusz W and Bumby, James R},
  year={2020},
  publisher={John Wiley \& Sons}
}

@article{trip2016internal,
  title={An internal model approach to (optimal) frequency regulation in power grids with time-varying voltages},
  author={Trip, Sebastian and B{\"u}rger, Mathias and De Persis, Claudio},
  journal={Automatica},
  volume={64},
  pages={240--253},
  year={2016},
  publisher={Elsevier}
}

@article{zhou2008improved,
  title={Improved swing equation and its properties in synchronous generators},
  author={Zhou, Jun and Ohsawa, Yasuharu},
  journal={IEEE Transactions on Circuits and Systems I: Regular Papers},
  volume={56},
  number={1},
  pages={200--209},
  year={2008},
  publisher={IEEE}
}

@article{zhong2010synchronverters,
  title={Synchronverters: Inverters that mimic synchronous generators},
  author={Zhong, Qing-Chang and Weiss, George},
  journal={IEEE transactions on industrial electronics},
  volume={58},
  number={4},
  pages={1259--1267},
  year={2010},
  publisher={IEEE}
}

@book{grainger1999power,
  title={Power system analysis},
  author={Grainger, John J},
  year={1999},
  publisher={McGraw-Hill}
}

@book{sauer2017power,
  title={Power system dynamics and stability: with synchrophasor measurement and power system toolbox},
  author={Sauer, Peter W and Pai, Mangalore A and Chow, Joe H},
  year={2017},
  publisher={John Wiley \& Sons}
}

@article{kundur2007power,
  title={Power system stability},
  author={Kundur, Prabha},
  journal={Power system stability and control},
  volume={10},
  number={1},
  pages={7--1},
  year={2007}
}

@article{brogliato2007dissipative,
  title={Dissipative systems analysis and control},
  author={Brogliato, Bernard and Lozano, Rogelio and Maschke, Bernhard and Egeland, Olav and others},
  journal={Theory and applications},
  volume={2},
  pages={2--5},
  year={2007},
  publisher={Springer}
}

@article{ghallab2025negative,
  title={Negative imaginary systems theory for nonlinear systems: A dissipativity approach},
  author={Ghallab, Ahmed G and Mabrok, Mohammed A and Petersen, Ian R},
  journal={IEEE Transactions on Automatic Control},
  year={2025},
  publisher={IEEE}
}

@article{willems1972dissipative,
  title={Dissipative dynamical systems part I: General theory},
  author={Willems, Jan C},
  journal={Archive for rational mechanics and analysis},
  volume={45},
  number={5},
  pages={321--351},
  year={1972},
  publisher={Springer}
}

@article{petersen2010feedback,
  title={Feedback control of negative-imaginary systems},
  author={Petersen, Ian R and Lanzon, Alexander},
  journal={IEEE Control Systems Magazine},
  volume={30},
  number={5},
  pages={54--72},
  year={2010},
  publisher={IEEE}
}

@article{kundur2004definition,
  title={Definition and classification of power system stability IEEE/CIGRE joint task force on stability terms and definitions},
  author={Kundur, Prabha and Paserba, John and Ajjarapu, Venkat and Andersson, G{\"o}ran and Bose, Anjan and Canizares, Claudio and Hatziargyriou, Nikos and Hill, David and Stankovic, Alex and Taylor, Carson and others},
  journal={IEEE transactions on Power Systems},
  volume={19},
  number={3},
  pages={1387--1401},
  year={2004},
  publisher={IEEE}
}

@article{yang2019distributed,
  title={Distributed stability conditions for power systems with heterogeneous nonlinear bus dynamics},
  author={Yang, Peng and Liu, Feng and Wang, Zhaojian and Shen, Chen},
  journal={IEEE Transactions on Power Systems},
  volume={35},
  number={3},
  pages={2313--2324},
  year={2019},
  publisher={IEEE}
}

@article{monshizadeh2018output,
  title={Output impedance diffusion into lossy power lines},
  author={Monshizadeh, Pooya and Monshizadeh, Nima and De Persis, Claudio and van der Schaft, Arjan},
  journal={IEEE Transactions on Power Systems},
  volume={34},
  number={3},
  pages={1659--1668},
  year={2018},
  publisher={IEEE}
}

@inproceedings{nishino2025equilibrium,
  title={Equilibrium-Independent Passivity of Power Systems Composed of Park Synchronous Generator Models},
  author={Nishino, Taku and Ishizaki, Takayuki},
  booktitle={2025 American Control Conference (ACC)},
  pages={748--753},
  year={2025},
  organization={IEEE}
}

@article{arghir2018grid,
  title={Grid-forming control for power converters based on matching of synchronous machines},
  author={Arghir, Catalin and Jouini, Taouba and D{\"o}rfler, Florian},
  journal={Automatica},
  volume={95},
  pages={273--282},
  year={2018},
  publisher={Elsevier}
}
\end{document}